\documentclass[10pt,twocolumn,twoside]{IEEEtran}

\usepackage{graphicx}
\usepackage{graphics} 
\usepackage{epsfig} 
\usepackage{amsmath} 
\usepackage{amssymb}  
\usepackage{amsfonts}
\usepackage{amsthm}
\usepackage{blkarray}
\usepackage{arydshln}
\usepackage{algorithmic}
\usepackage{algorithm}
\usepackage{url}
\usepackage{color}
\usepackage{tikz}
\usepackage[compress]{cite}
\usepackage[font=small]{caption}

\usepackage{enumitem}
\usepackage{caption}
\usepackage{thmtools}

\usetikzlibrary{matrix,arrows}
\newcommand{\algorithmicoutput}{\textbf{Output:}}
\newcommand{\OUTPUT}{\item[\algorithmicoutput]}

\newtheorem{theorem}{\bf Theorem}

\newtheorem{definition}{\bf Definition}

\newtheorem{proposition}{\bf Proposition}

\newtheorem{assumption}{\bf Assumption}

\declaretheorem[style=definition,qed=$\square$]{remark}
\declaretheorem[style=definition,qed=$\square$]{example}

\DeclareMathOperator*{\argmin}{arg\,min}

\newcommand*{\LONGVERSION}{}
\newcommand{\upd}{}

\begin{document}


\author{
    \IEEEauthorblockN{Takashi Tanaka\IEEEauthorrefmark{1}, Farhad Farokhi\IEEEauthorrefmark{2}, C\'edric Langbort\IEEEauthorrefmark{3}}\\
    \IEEEauthorblockA{\IEEEauthorrefmark{1}LIDS, Massachusetts Institute of Technology, USA
    \texttt{ttanaka@mit.edu}} \\
    \IEEEauthorblockA{\IEEEauthorrefmark{2}Dept. of Electrical and Electronic Eng., University of Melbourne, Australia
    \texttt{ffarokhi@unimelb.edu.au}} \\
    \IEEEauthorblockA{\IEEEauthorrefmark{3}Coordinated Science Lab., University of Illinois at Urbana-Champaign, USA
    \texttt{langbort@illinois.edu}}
}

\title{
Faithful Implementations of Distributed Algorithms and Control Laws\thanks{An early version of this article was presented at the 52nd IEEE Conference on Decision and Control~\cite{TanakaCDC2013}.} }
\maketitle

\begin{abstract} 
When a distributed algorithm must be executed by strategic agents with misaligned interests, a social leader needs to introduce an appropriate tax/subsidy mechanism to incentivize agents to faithfully implement the intended algorithm so that a correct outcome is obtained. 
We discuss the incentive issues of implementing economically efficient distributed algorithms using the framework of indirect mechanism design theory. 
In particular, we show that indirect Groves mechanisms are not only sufficient but also necessary to achieve incentive compatibility. This result can be viewed as a generalization of the Green-Laffont theorem to indirect mechanisms. 
Then we introduce the notion of asymptotic incentive compatibility as an appropriate solution concept to faithfully implement distributed and iterative optimization algorithms.
We consider two special types of optimization algorithms: dual decomposition algorithms for resource allocation and average consensus algorithms. 
\end{abstract}

\section{Introduction}
\label{secintro}

In this paper, we consider a society comprised of a single leader and $N$ followers. The leader  makes a social decision $z\in \mathcal{Z}$,\footnote{For ease of presentation, we assume $\mathcal{Z} \subset \mathbb{R}^{n_z}$.} which incurs cost $v_i(z;\theta_i)$ to the $i$-th follower. 
For every $i=1,2,\cdots,N$, assume that cost function $v_i(z;\theta_i)$ has a known parametric model  while parameters $\theta_i\in\Theta_i$ are private. (For instance, $v_i(z;\theta_i)$ can be a $10$-th order polynomial of a scalar variable $z$ whose coefficients $\theta_i\in\mathbb{R}^{10}$ are private.)
The leader desires to make a social decision $z^*$ that minimizes the sum of the followers' individual costs; 
\begin{equation}
\label{zstar}
z^*\in \argmin_{z\in \mathcal{Z}} \sum_{i=1}^N v_i(z;\theta_i).
\end{equation}
A social decision satisfying (\ref{zstar}) is said to be \emph{(economically) efficient}.
{\upd Efficient decision-making requires distributed algorithms involving leader-follower communication, since the leader has no access to the private parameters.}
A challenge here is that such decision mechanisms must be designed so that self-interested followers are  given no incentive to manipulate the algorithm in an effort to minimize their individual costs. 
Note that this requirement is different from the fault resilience requirement considered in, for instance, the Byzantine generals problem \cite{lamport1982}. We are interested in decision mechanisms in which  the followers could manipulate the result, but they choose not to do so.

Roughly speaking, the task of the leader is to design a game that produces an efficient decision
$z^*$
as a consequence of game-theoretic equilibrium strategies of  the followers.
Designing such games systematically in various multi-agent decision making situations (e.g., auctions, elections, resource allocations), frequently using a carefully designed tax/subsidy rule,  is a subject of interest in the Economics literature under the umbrella of \emph{mechanism design} theory. 
{\upd Developments since the 1970s in mechanism design theory have resulted in a  rich and established discipline; basic information about mechanism design theory can be found in, e.g.,~\cite{RefWorks:260,myerson1997,RefWorks:83,RefWorks:298,RefWorks:100,krishna2009}.}
One of the best-known positive results is the Groves mechanism, which provides clear guidelines to design tax/subsidy rules incentivizing  the followers to be collaborative in the process of computing efficient decisions. 

Recently, the theory of mechanism design has been applied to various engineering and computer scientific problems. These applications have raised new challenges to the traditional mechanism design theory, in term of computational difficulties (e.g., combinatorial auctions \cite{RefWorks:108}\cite{RefWorks:299}, job scheduling \cite{archer2001}) and communication difficulties (e.g., inter-domain routing \cite{feigenbaum2005}). 
{\upd
For instance,  the standard Groves mechanism becomes computationally intractable if the optimization problem (\ref{zstar}) is NP-hard.
In such cases, the goal of mechanism design must be set alternatively to incentivize strategic followers to act faithfully in a (computationally feasible) algorithm that only approximates an optimal solution.
It turns out that this task is not straightforward, since mechanisms that naively approximate the Groves mechanism are in general \emph{not} ``approximately incentive compatible" \emph{at all} \cite{RefWorks:235}\cite{RefWorks:236}.
This implies that a fundamental departure from the Groves mechanism is inevitable when the underlying optimization problem is computationally hard.} The interplay between incentives, computation, and communication now forms the field of \emph{algorithmic mechanism design (AMD)}  \cite{RefWorks:236}\cite{RefWorks:288}. 
{\upd It should be noted that several recent papers discuss similar ideas regarding game designs for distributed optimization/control without referring to the AMD theory explicitly \cite{marden2008,RefWorks:237,li2013}.
However, their intrinsic connections to the mechanism design theory should be clarified to facilitate the further developments beyond their current problem-specific nature.
}

The simplest approach that the leader can take to obtain a solution  in  (\ref{zstar}) is to incentivize followers to report their private parameters $\theta_i$ truthfully, so that the leader can solve the optimization problem using the central computer.
This particular type of decision making procedure, called \emph{direct revelation mechanism}, has been the main focus in the mechanism design literature. 
{\upd In many cases, this concentrated interest can be justified by the \emph{revelation principle} \cite{myerson1997}, which proves the existence of an incentive compatible direct mechanism for every incentive compatible indirect mechanism, and thus guarantees no loss of generality with focusing only on direct mechanisms. However, as recognized in the AMD literature, there are many situations in which the revelation principle should not be naively relied on:}
\begin{itemize}[leftmargin=*]
\item[i)] In a direct revelation mechanism, the leader is solely responsible for solving a possibly large-scale optimization problem. {\upd Even if there exist more practical distributed algorithms, direct mechanisms do not allow distributed implementations of algorithms by the followers;}
\item[ii)] Reporting $\theta_i$ is not a trivial action: It may be difficult for the followers to identify their own cost function; 
\item[iii)] Reporting $\theta_i$ means a complete loss of privacy;
\item[iv)] Direct and secure communication links between the leader and the followers are not always available\footnote{This issue is raised in \cite{monderer1999}. See also Remark \ref{remchangerule}.};
\item[v)] In homogeneous environments (e.g., internet), it might be difficult to establish a leader\footnote{Since we will always assume that a leader exists in this paper, the item v) is beyond our scope.  However, we note that there are a few successful mechanism design examples under such environments, including  multi-cast cost sharing  \cite{feigenbaum2001} and interdomain routing \cite{feigenbaum2005}. Some important results in distributed algorithmic mechanism design (DAMD) as of 2002 are summarized in a review paper \cite{feigenbaum2002}.}.
\end{itemize}

{\upd To resolve these issues, it is invaluable to develop a general guideline to design distributed algorithms that induce \emph{truthful actions} by the followers. This requirement is far more general than the one in the direct mechanism regime, where only \emph{truthful reports} are considered.} 
The possibility of such generalization is foreseen by several encouraging results. In  \cite{RefWorks:233}, it is shown that a natural generalization of the Groves mechanism to the indirect mechanisms (referred to as indirect Groves mechanisms in this paper) implements a socially optimal set of strategies in \emph{ex-post} Nash equilibria. The idea is employed in \cite{RefWorks:300}, where several concrete distributed algorithms are shown to be faithfully implementable by strategic agents. Coordination of strategic agents in a dynamic decision making process is considered in~\cite{cavallo2012}.

The approach of \cite{RefWorks:233} is particularly attractive since, unlike many results in AMD that are problem-specific, the result there is applicable to a wide range of distributed computation and communication protocols. 
In this paper, we pursue the same direction of research and make several additional observations that are essential especially when deploying these results in distributed numerical optimization and control problems.
We present the following technical contributions in this paper.

\subsubsection{Necessity of indirect Groves mechanisms} We prove in Theorem \ref{theoconverse} that incentive compatible indirect mechanisms must be in the class of indirect Groves mechanisms whenever the space $\Theta_i$ for some $i$ is rich enough to parametrize all quadratic cost functions. This is an extension of the Green-Laffont Theorem \cite{green1977} to the indirect  mechanism setting.
\subsubsection{Asymptotic incentive compatibility}
We introduce this solution concept to justify the use of approximated Groves taxes to incentivize followers to ``act right" in a wide class of optimization algorithms, including continuous optimization algorithms.
Due to the nature of the continuous optimization, the exact solution cannot be obtained in finite time and hence, Groves taxes must be inevitably approximated as well. 
However, it has not been fully discussed in the literature  whether the use of approximated Groves mechanisms in this context is justifiable or not.
We argue that the use of approximated Groves mechanisms is justifiable whenever we have an iterative distributed algorithm, which can be iterated as many times as we wish, and we can compute approximated Groves taxes from its output that diminish the followers' incentives for cheating to an arbitrary small $\epsilon$. We believe such a situation is satisfactory to convince followers to ``act right" in the algorithm, and hence is a practical solution concept. We name this solution concept ``asymptotic incentive compatibility."

{\upd As mentioned earlier, the issue of approximating Groves taxes is well studied in the AMD literature. However, our focus in item 2) above is different. 
 In the AMD literature, the research focus is almost exclusively on discrete optimization with approximation threshold strictly greater than zero (in the language of \cite{papadimitriou1995}).
In such cases, the research focus must be on non-Groves mechanisms, since Groves mechanisms are computationally impractical. 
On the other hand, our focus is still on the (indirect) Groves mechanisms.
We consider their applications to continuous optimization problems and clarify in what sense a mechanism can be ``incentive compatible" in those cases. }


This paper is organized as follows. We start with a motivating example in Section \ref{secresourceallocation}.  Section \ref{secimd} formally introduces the framework of indirect mechanism design. Section \ref{secaic} develops the notion of asymptotic incentive compatibility. In Sections \ref{secdd} and \ref{secaverage}, we discuss faithful implementations of dual decomposition algorithms and average consensus algorithms. Section \ref{secdiscussion} contains some additional discussion and conclusions.

\section{Motivating Examples}
\label{secresourceallocation}
Consider a resource allocation problem of the form
\begin{subequations}
 \label{originalprob}
\begin{align}
\min &\;\; \sum_{i=1}^N v_i(z_i;\theta_i), \label{originalprob1}\\
\text{s.t.} &\;\; Rz=c.  \label{originalprob2}
\end{align}
\end{subequations}
A vector $z=(z_1, \cdots, z_N) \in \mathbb{R}^{n_{z_1}}\times \cdots \times \mathbb{R}^{n_{z_N}}$ is a concatenation of the social decision variables, 
and $R=[R_1 \; \cdots \; R_N]$.
{\upd Define $\mathcal{Z}$ to be the set of all $z$ such that $Rz=c$.}

Let $L(z,\lambda)= \sum_{i=1}^N v_i(z_i;\theta_i)+\lambda^{\top}(Rz-c)$ be the Lagrangian with a Lagrange multiplier $\lambda$. 
The primal-dual optimal solution $(z^*,\lambda^*)$ is a saddle point of $L(z,\lambda)$, and assuming that cost functions are strictly convex, the saddle point value $L^*$ is equal to the optimal value of (\ref{originalprob}); see~\cite{boyd2004convex}. The following iterations are guaranteed to converge to $(z^*,\lambda^*)$ if the step size $\gamma$ is chosen to be sufficiently small~\cite{boyd2003subgradient}:
\begin{subequations}
\label{ahu}
\begin{align}
&\hat{z}_i^k=\argmin_{\hat{z}_i} \left( v_i(\hat{z}_i,\theta_i)+(\lambda^{k-1})^{\top}R_i\hat{z}_i \right) \label{ahu1} \\
&\lambda^k=\lambda^{k-1}+ \gamma(R\hat{z}^k-c). \label{ahu2}
\end{align}
\end{subequations}

Notice that the above algorithm has an attractive form for a distributed implementation since (\ref{ahu1}) can be executed by the followers. This type of parallelization is known as \emph{dual decomposition}.
By increasing the number of iterations, the optimal social decision can be approximated with an arbitrary accuracy, provided that the followers faithfully implement (\ref{ahu1}). 

What kind of side payment (tax/subsidy) mechanism do we need to incentivize the followers to execute (\ref{ahu1})? Let us make a first attempt. 
Suppose that the leader introduces the following format of auction mechanism:
\begin{itemize}[leftmargin=*]
\item[] {\bf Step 1}: Each follower has some initial value $z_i^0$, and the leader has some initial value $\lambda^0$.
\item[] {\bf Step 2}: Run iteration (\ref{ahu}) until it reaches a convergence to $(z_1^*,\cdots, z_N^*, \lambda^*)$.
\item[] {\bf Step 3}: The leader determines the allocation according to $(z_1^*,\cdots, z_N^*)$, and each player makes a payment $p_i^0=\lambda^* z_i^*$ to the leader.
\end{itemize}

In microeconomic theory, $\lambda^*$ is known as the ``market-clearing price" under which demand $\sum_{i=1}^N R_iz_i$ and supply $c$ are balanced.
The above mechanism employs a particular type of tax rule $p_i^0$ which is very natural: the tax imposed on the $i$-th player is calculated by the share he has won times the market-clearing price. 

Unfortunately, the above auction mechanism is \emph{not} incentive compatible. It is easy to demonstrate that it is vulnerable to strategic manipulations. Suppose $N=2$, $c=1$, $R_i=1$, and $v_i(z_i;\theta_i)=(z_i-1)^2$ for $i=1,2$. If both players follow the suggested algorithm, the iteration reaches the optimal solution $(z_1^*, z_2^*, \lambda^*)=(1/2, 1/2, 1)$, which brings a net cost of $v_i(z_i^*;\theta_i)+\lambda^*z_i^*=3/4$ to each player. 
Now, suppose that player $1$ bids $\hat{z}_1^k=1/3$ at every iteration (instead of executing (\ref{ahu1}) faithfully). Then it can be  shown that the iteration arrives at a different fixed point $(z_1', z_2', \lambda')=(1/3, 2/3, 2/3)$. This result brings net cost of $v_1(z_1';\theta_1)+\lambda' z_1'=2/3$ to player $1$, which is less than $3/4$. Hence, player 1 is indeed better off by deviating from (\ref{ahu1}).

We also note that, when followers are ``price-takers" (which is the case, for instance, when every follower has sufficiently small market power and the price cannot be affected by his sole action), it makes sense to assume that each follower executes (\ref{ahu1}) in an effort to minimize his own net cost. 
However, many realistic markets are oligopolistic, in which a stakeholder agent knows that his sole action has a certain effect on the market-clearing price \cite{RefWorks:238}. In this case, he might be better off by ``exercising market power" rather than following (\ref{ahu1}) as shown in the above example.
Analyzing strategic bidding in a given auction mechanism (as in \cite{hao1999,kian2005,johari2004})  is an important topic. In this paper, however, we are more interested in designing mechanisms in which strategic manipulation by a follower brings no benefit to him.

The primal-dual algorithm considered in this section is just a motivating example. The result of the next section is applicable to a much more general class of distributed algorithms.

\section{Indirect Mechanism Design}
\label{secimd}
\subsection{Framework}

\begin{figure}[t]
\centering
 \begin{tikzpicture}[descr/.style={fill=white,inner sep=2.5pt}]
 \matrix(m) [matrix of math nodes, row sep=3em,
 column sep=3em]
 { \Theta & & \mathcal{Z} \times \mathcal{P} \\ & \mathcal{S} & \\ };
 \path[->]
 (m-1-1) edge node[auto] {$ f=(\zeta,\pi) $} (m-1-3)
         edge node[below left=2pt] {$ \sigma $} (m-2-2)
 (m-2-2) edge node[below right=2pt] {$ g=(g_\zeta,g_\pi) $} (m-1-3);
 \end{tikzpicture}

\captionsetup{singlelinecheck=off}
\caption[]{\label{commdiag} Framework of indirect mechanisms.
{\upd
  \begin{itemize}
\item[] $\Theta=\Theta_1\times\cdots\times \Theta_N$: The space of private parameters.
\item[] $\mathcal{P}\subseteq \mathbb{R}^N$: The space of tax values $p=(p_1,\cdots,p_N)$.
\item[] $\mathcal{Z}\subseteq \mathbb{R}^{n_z}$: The space of social decisions $z$.
\item[] $\mathcal{S}=\mathcal{S}_1\times\cdots\times \mathcal{S}_N$: The space of actions.
\item[] $\zeta:\Theta\rightarrow \mathbb{R}^{n_z}$: Decision rule.
\item[] $\pi:\Theta\rightarrow \mathcal{P}$: Tax rule.
\item[] $\sigma_i:\Theta_i\rightarrow \mathcal{S}_i$: Strategy function.
\item[] $f:\Theta\rightarrow \mathcal{Z}\times \mathcal{P}$: Social choice function.
\item[] $g:\mathcal{S}\rightarrow \mathcal{Z}\times \mathcal{P}$: Outcome function.
  \end{itemize}
}
}
\vspace{-3ex}
\end{figure}
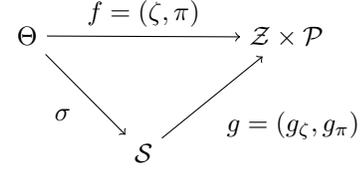

The diagram in Fig. \ref{commdiag} shows an abstract framework for indirect mechanisms.
Let $\Theta=\Theta_1\times \cdots \times \Theta_N$ be the space of private parameters. A function $\zeta: \Theta\rightarrow \mathbb{R}^{n_z}$ is called a \emph{decision rule}. For fixed sets $\Theta$ and the space $\mathcal{Z}\subseteq  \mathbb{R}^{n_z}$ of social decisions, a decision rule is said to be \emph{efficient}  if $\zeta(\Theta)\subseteq \mathcal{Z}$ and
$
\sum_{i=1}^N v_i(\zeta(\theta);\theta_i) \leq \sum_{i=1}^N v_i(z;\theta_i)
$
for all $\theta\in \Theta$ and for all $z\in \mathcal{Z}$.
Let us also introduce a \emph{tax rule} $\pi:\Theta\rightarrow \mathcal{P}$ where $\mathcal{P} \subset \mathbb{R}^N$ is the space of tax values assigned to the followers. The pair $f=\left(\zeta,\pi\right)$, $f:\Theta\mapsto \mathbb{R}^{n_z} \times \mathbb{R}^N$ is called a \emph{social choice function}.

When the leader designs a decision making mechanism, the \emph{action space} $\mathcal{S}=\mathcal{S}_1\times \cdots \times \mathcal{S}_N$ must be specified, where $\mathcal{S}_i$ can be thought of as the space of all possible programming codes that the $i$-th follower can potentially execute in the distributed computation (for instance, executing (\ref{ahu1}) is a valid action in the algorithm considered in Section \ref{secresourceallocation}, while bidding $\hat{z}_1^k=1/3$ at every step is another).
The $i$-th follower with private parameter $\theta_i$ determines his actions according to the \emph{strategy function} $\sigma_i:\Theta_i \rightarrow \mathcal{S}_i$.
Outputs of the followers' algorithms are processed by the leader's algorithm called the \emph{outcome function} $g:\mathcal{S}\rightarrow \mathcal{Z} \times \mathcal{P}$, which determines a social decision and tax values.
\begin{remark}
For example, the above formulation can express  the following abstract model of multi-stage interactions between the leader and the followers.
At each stage (indexed by $k=1,2,\cdots, n$), the leader broadcasts his current computational output $\eta_L^k$ to the followers. Each follower transmits his current computational output $\eta_i^k$ to the leader and other followers. Assume that the leader and followers can be modeled as a state-based computer with the internal state $\xi_L^k$ and $\xi_i^k, i=1,\cdots, N$ respectively. Given initial states $\xi_L^0,\xi_i^0$ and $\eta_L^0, \eta_i^0, i=1,\cdots,N$, the state evolves according to:
\begin{subequations}
\label{algall}
\begin{align}
\xi_L^k &= G_L^k (\xi_L^{k-1}, \eta_1^k, \cdots, \eta_N^k);  && \eta_L^k = H_L^k (\xi_L^k), \label{algL}\\
\xi_i^k &= G_{i,\theta_i}^k(\xi_i^{k-1},\eta_L^{k-1}, \eta_{-i}^{k-1}); && \eta_i^k = H_{i,\theta_i}^k(\xi_i^k), \label{algFi}
\end{align}
\end{subequations}
for $k=1,2,\cdots,n$. 
Finally, we require that
$
\eta_L^n=H_L^n(\xi_L^n)=(\zeta(\theta),\pi(\theta)),
$
which will be the value of the social choice.
In this communication model, the action of the $i$-th follower is the sequence of functions in (\ref{algFi}), i.e., 
$
\sigma_i(\theta_i)=\{(G_{i,\theta_i}^k, H_{i,\theta_i}^k): k=1,2,\cdots, n\}
$
parametrized by his type $\theta_i$. The outcome function is defined by the sequence of functions in (\ref{algL}), i.e.,
$
g=\{(G_L^k, H_L^k): k=1,2,\cdots, n\}.
$
 For a  practical implementation, $n$ must be a finite number.
\end{remark}

Formally, a \emph{mechanism} is a triplet $M=(\sigma, \mathcal{S}, g)$ of a strategy function $\sigma$, the space of collective actions $\mathcal{S}$, and an outcome function $g$.
Notice that a mechanism $M$ suggests followers to employ a particular strategy function specified by $\sigma$, but it is followers' choice to be faithful or not; followers are allowed to take any actions in $\mathcal{S}$.

A particular case with $\mathcal{S}=\Theta$, $g=f$, and $\sigma=Id$ (i.e., identity map) is called a \emph{direct mechanism}, in which followers are asked to report their types $\theta_i$ to the leader directly. A mechanism $M=(\sigma,\mathcal{S}, g)$  is said to be \emph{dominant strategy incentive compatible} if implementing the suggested action $\sigma(\theta)=(\sigma_1(\theta_1), \cdots, \sigma_N(\theta_N))\in \mathcal{S}$ constitutes dominant strategies among the followers. 
However, this requirement turns out to be often difficult to attain in indirect mechanism design settings. Hence, we employ a weaker notion of incentive compatibility.

\begin{definition} 
\label{def:feasible}
A mechanism $M=(\sigma, \mathcal{S}, g)$ is said to be \emph{single fault tolerant} if for any $i$,  
$s_i\in \mathcal{S}_i$, and $\theta\in\Theta$, the mechanism produces a feasible outcome 
$g(s_i,\sigma_{-i}(\theta_{-i}))\in \mathcal{Z}\times \mathcal{P}$. 
\end{definition}
\begin{definition}
\label{defic}
A mechanism  $M=(\sigma, \mathcal{S}, g)$  is said to be \emph{incentive compatible} if $\forall i, \forall s_i\in\mathcal{S}_i, \forall \theta\in \Theta$,
\begin{align*}
\hspace{+.7in}&\hspace{-.7in}v_i(g_\zeta(\sigma_i(\theta_i),\sigma_{-i}(\theta_{-i}));\theta_i)
+g_{\pi_i}(\sigma_i(\theta_i),\sigma_{-i}(\theta_{-i})) \\
&\leq v_i(g_\zeta(s_i,\sigma_{-i}(\theta_{-i}));\theta_i)
+g_{\pi_i}(s_i,\sigma_{-i}(\theta_{-i})).
\end{align*}
In this case, the mechanism is said to \emph{implement} a social choice function $f=g\circ \sigma$ in \emph{ex-post Nash equilibria}\footnote{The term \emph{ex-post} is commonly used to mean that $\sigma_i(\theta_i)$ is the best strategy even without knowing $\theta_{-i}$.  See \cite{RefWorks:288}, Section 9. }. 
\end{definition}
{\upd Intuitively, single fault tolerance requires the mechanism to make a valid social decision (if not optimal) even if at most one follower did not implement suggested strategies faithfully.
Incentive compatibility requires that no follower is incentivized to deviate from the suggested strategy if all other followers faithfully implement suggested strategies.}

\subsection{Indirect Groves mechanism}
\label{secgroves}
{\upd
In this paper, we focus on \emph{efficient} distributed algorithms (i.e., those that minimize social cost), and designing a tax rule that induces the followers' faithful actions in such algorithms.}
The question is rephrased as follows: Given a pair $(g_\zeta, \sigma)$ such that $g_\zeta\circ \sigma$ is efficient, how can we design $g_\pi$ so that $M=(\sigma,\mathcal{S},g)$ is incentive compatible?
\begin{definition}
\label{defindirectgroves}
A mechanism $M=(\sigma, \mathcal{S}, g)$ is said to be in the class of \emph{indirect Groves mechanisms} if, for every $i\in\{1,\cdots, N\}$, there exists a function $k_i: \mathcal{S}_{-i}\rightarrow \mathbb{R}$ satisfying:
\begin{itemize}[leftmargin=*]
\item For every $ \theta_{-i}\in\Theta_{-i}$ and $s_i \in \sigma_i(\Theta_i)$, the tax rule is
\begin{align}
g_{\pi_i}(s_i,&\sigma_{-i}(\theta_{-i})) \nonumber \\
& =  \sum_{j\neq i}  v_j(g_\zeta(s_i,\sigma_{-i}(\theta_{-i}));\theta_j)+k_i(\sigma_{-i}(\theta_{-i})) \label{grovestaxequation}
\end{align}
\item For every $ \theta_{-i}\in\Theta_{-i}$ and $s_i \in \mathcal{S}_i \setminus  \sigma_i(\Theta_i)$, the tax rule satisfies
\begin{align}
g_{\pi_i}(s_i&,\sigma_{-i}(\theta_{-i})) \nonumber  \\
& \geq  \sum_{j\neq i}  v_j(g_\zeta(s_i,\sigma_{-i}(\theta_{-i}));\theta_j)+k_i(\sigma_{-i}(\theta_{-i})). \label{grovestaxineq}
\end{align}
\end{itemize}
\end{definition}

\begin{theorem}\label{theoic}(Sufficiency)
A single fault tolerant mechanism $M=(\sigma,\mathcal{S},g)$ with an efficient decision rule $g_\zeta\circ \sigma$ is incentive compatible if it is in the class of indirect Groves mechanisms.
\end{theorem}
\begin{proof}
Suppose on contrary that
\begin{align*}
\hspace{+.7in}&\hspace{-.7in}v_i(g_\zeta(\sigma_i(\theta_i), \sigma_{-i}(\theta_{-i}));\theta_i)+g_{\pi_i}(\sigma_i(\theta_i),\sigma_{-i}(\theta_{-i})) \\
&\;\; > v_i(g_\zeta(s_i, \sigma_{-i}(\theta_{-i}));\theta_i)+g_{\pi_i}(s_i,\sigma_{-i}(\theta_{-i})). 
\end{align*}
for some $i, s_i\in \mathcal{S}_i, \theta\in\Theta$.
The equality (\ref{grovestaxequation}) is applicable on the left hand side, while  (\ref{grovestaxequation}) or (\ref{grovestaxineq}) is used on the right hand side depending on $s_i$. 
In both cases, the above inequality implies
$$
\sum_{i=1}^N v_i(g_\zeta \circ \sigma(\theta); \theta_i) 
> \sum_{i=1}^N v_i(g_\zeta(s_i,\sigma_{-i}(\theta_{-i})); \theta_i). 
$$
Since $M$ is single fault tolerant, $g_\zeta(s_i,\sigma_{-i}(\theta_{-i}))\in \mathcal{Z}$. Hence, this is a contradiction to the efficiency of $g_\zeta\circ s$.
\end{proof}
Theorem \ref{theoic} is due to \cite{RefWorks:233}. A less trivial fact is that the converse of Theorem \ref{theoic} also holds when each agent's private parameter space $\Theta_i$ is rich enough.
\begin{assumption}
\label{asmp1}
For every $i$ and every quadratic function $q(\cdot):\mathcal{Z}\rightarrow \mathbb{R}$, there exists $\theta_i\in\Theta_i$ such that $q(\cdot)=v_i(\cdot;\theta_i)$.
\end{assumption}

\begin{theorem} 
\label{theoconverse}(Necessity)
Suppose Assumption \ref{asmp1} holds. A single fault tolerant mechanism $M=(\sigma, \mathcal{S}, g)$ with an efficient decision rule $g_\zeta \circ \sigma$  is incentive compatible only if it is in the class of indirect Groves mechanisms.
\end{theorem}
\begin{proof}
Complete proof can be found in Appendix \ref{apendnec}. The basic idea of the proof is attributed to the celebrated result by Green and Laffont \cite{green1977}. The proof for Case 1 is a modification of Theorem 10.4.3 in \cite{RefWorks:100}.
\end{proof}

\begin{remark} Unlike the Groves mechanism in the direct mechanism design, the indirect Groves mechanism does not generally implement the desired algorithm in a dominant strategy. 
Indeed, it was shown by Proposition 9.23 in \cite{RefWorks:288} that an indirect mechanism is dominant strategy incentive compatible only if every map $\sigma_i:\Theta_i\rightarrow \mathcal{S}_i$ is surjective.
We will see a concrete example of this fact in Example \ref{solconcept}.  
\end{remark}
\begin{remark}
\label{remchangerule}
The set $\mathcal{S}$ can be extremely rich, since it is the space of all programs that can be executed by the followers during the course of the algorithm. For instance, followers are allowed to write a code to learn about other followers during the algorithm to make future decisions. However, as an implicit premise for Theorem \ref{theoic} and \ref{theoconverse}, we must preclude the followers' ability to ``hack the rule of the game." For instance, the intended algorithm must be securely announced to every follower without strategic interventions by other followers. Similarly, the tax value must be securely computed based on the formula  (\ref{grovestaxequation}) and (\ref{grovestaxineq}) without a danger of manipulation. Such interventions are actually possible if a follower has an opportunity to modify other players' messages  \cite{monderer1999}. For the same reason, followers are not allowed to drop out of the game in midstream to escape from a punitive tax.
\end{remark}
\begin{remark}
Direct communication links between the leader and the followers are not constantly required during the course of the algorithm.
For instance, the main body of the average consensus algorithm in Section \ref{secaverage} requires only peer-to-peer communications among neighboring followers, but followers still cannot be better off by cheating. However, in order to satisfy the requirement of the previous remark, secure communication links between the leader and followers are assumed at the initial phase (to announce the algorithm) and at the final phase (to calculate taxes securely).
\end{remark}

\subsection{Individual rationality, Budget balance}
{\upd A mechanism is said to be \emph{individually rational}~\cite{RefWorks:100} if the net cost $v_i(g_\zeta(s);\theta_i)+g_{\pi_i}(s)$ is non-positive for every $i$.
This is a basic requirement for a mechanism that does not incentivize the followers to quit the mechanism, when quitting the mechanism is cost-free for them.

A mechanism is said to be \emph{budget balanced} (resp. \emph{weakly budget balanced})~\cite{RefWorks:100} if the tax income $\sum_{i=1}^N p_i$ is zero (resp. non-negative).

Unfortunately, there may not exist a mechanism that simultaneously satisfies (1) efficiency, (2) incentive compatibility, (3) individual rationality, and (4) budget balance.
The aforementioned result by Green and Laffont \cite{green1977} shows that the only efficient direct mechanisms that are dominant strategy incentive compatible are Groves mechanisms. 
This observation allows us to construct a simple example in which no efficient mechanism simultaneously achieves dominant strategy incentive compatibility, weak budget balance, and individual rationality. 
In a Bayesian setting, \cite{myerson1983efficient} demonstrated that there exists a simple exchange environment in which no (\emph{ex post}) efficient mechanism is simultaneously (Bayes-Nash) incentive compatible, weakly budget balanced, and (\emph{ex interim}) individually rational.
This result was generalized in \cite{krishna1998efficient} using the revenue maximization principle.
We also note that a recent study \cite{marden2009} of a particular indirect mechanism shows that no budget balanced mechanism implements efficient decisions in Nash equilibrium.
 }

\section{Asymptotic incentive compatibility}
\label{secaic}
{\upd In this section, we generalize Theorem \ref{theoic} so that it is applicable to \emph{approximately} efficient decision rules. 
This is an important generalization, since in many realistic cases, 
social decision must be made upon the result of iterative numerical optimizations over continuous decision variables that, if terminated at some finite step, only returns an approximate solution.} In such cases, incentives may even be needed to guarantee not only that the suggested algorithm is implemented, but also that the actions taken by the followers lead it to converge.
Define $\mathrm{dist}(z',\mathcal{Z})=\inf_{z\in\mathcal{Z}} \|z'-z\|_2 $, and let $\text{Proj}(z';\mathcal{Z})$ be the projection of $z'$ onto $\mathcal{Z}$.
\begin{definition}
\label{defasymeff}
For every $n\in\mathbb{N}$, let $\zeta^n:\Theta\rightarrow {\upd \mathbb{R}^{n_z}}$ be a decision rule. A sequence of decision rules $\{\zeta^n\}_{n\in\mathbb{N}}$ is said to be \emph{asymptotically efficient} if $\mathrm{dist}(\zeta^n(\theta),\mathcal{Z})\rightarrow 0$ as $n\rightarrow \infty$ and for every $\epsilon>0$, there exists $n_0\in \mathbb{N}$ such that for all $n\geq n_0$,
$$
\sum_{i=1}^N v_i(\zeta^n(\theta);\theta_i)\leq
\sum_{i=1}^N v_i(z;\theta_i)+\epsilon,
\;\; \forall \theta\in \Theta, \; \forall z\in {\upd \mathcal{Z}}.
$$
\end{definition}

\begin{definition}
\label{defaic}
A sequence of mechanisms $\{M^n\}_{n\in\mathbb{N}}$, $M^n=(\sigma^n,\mathcal{S}^n,g^n)$, is said to be \emph{asymptotically incentive compatible} if for every $\epsilon>0$, there exists $n_0\in \mathbb{N}$ such that for $n\geq n_0$,
\begin{align*}
\hspace{.7in}&\hspace{-.7in}v_i\left(g_\zeta^n(\sigma_i^n(\theta_i),\sigma_{-i}^n(\theta_{-i})));\theta_i\right)
+g_{\pi_i}^n (\sigma_i^n(\theta_i),\sigma_{-i}^n(\theta_{-i})) \\
&\hspace{-.2in}\leq  v_i\left(g_\zeta^n (s_i^n,\sigma_{-i}^n(\theta_{-i}));\theta_i\right)
+g_{\pi_i}^n (s_i^n,\sigma_{-i}^n(\theta_{-i})) +\epsilon 
\end{align*}
$\forall i, \forall s_i^n\in\mathcal{S}_i^n, \forall \theta\in \Theta$. In this case, the mechanism is said to \emph{asymptotically implement} a social choice function $f:=\lim_{n\rightarrow \infty} g^n\circ \sigma^n$ in ex-post Nash equilibria, if the limit exists.
\end{definition}

\begin{remark}
{\upd For iterative algorithms, $n$ can be understood as the number of iterations before termination. 
Since this only gives an approximation of efficient social decisions, $M^n$ may not be incentive compatible for a fixed $n\in\mathbb{N}$. } However, as $n\rightarrow \infty$, $M^n$ provides every follower a diminishing incentive to deviate from the suggested  algorithm. Without loss of generality, we assume that taxes are paid after the algorithm has terminated.
\end{remark}

The next proposition presents a sequence of mechanisms $\{M^n\}_{n\in\mathbb{N}}$ motivated by the Groves mechanism by which the followers' incentive to deviate from the intended algorithm can be made arbitrary small.
\begin{proposition} \label{tho:general}
Let $\{M^n\}_{n\in\mathbb{N}}$, $M^n=(\sigma^n,\mathcal{S}^n,g^n)$, be a sequence of single fault tolerant mechanisms such that $\{g_\zeta^n \circ \sigma^n\}_{n\in\mathbb{N}}$ is asymptotically efficient. 
If the payment rule is
\[ g_{\pi_i}^n(s^n)=\sum_{j\neq i} v_j(g_\zeta^n(s^n);\theta_j)+k_i(s_{-i}^n)\;\; \forall s^n\in \mathcal{S}^n, \]
then $\{M^n\}_{n\in\mathbb{N}}$ is asymptotically incentive compatible.
\end{proposition}
\begin{proof}
Suppose there exist {\upd $\epsilon>0$}, $i$, $\theta\in\Theta$, a sequence of strategies $s^n_i\in\mathcal{S}_i^n$, and a subsequence $\{n_\ell\}_{\ell\in\mathbb{N}}$ such that
\begin{align*}
\hspace{.7in}&\hspace{-.7in}v_i(g_\zeta^{n_\ell}(\sigma_i^{n_\ell}(\theta_i),\sigma_{-i}^{n_\ell}(\theta_{-i})));\theta_i)
+g_{\pi_i}^{n_\ell} (\sigma_i^{n_\ell}(\theta_i),\sigma_{-i}^{n_\ell}(\theta_{-i})) \\
&\hspace{-.4in}>  v_i(g_\zeta^{n_\ell} (s_i^{n_\ell},\sigma_{-i}^{n_\ell}(\theta_{-i}));\theta_i)
+g_{\pi_i}^{n_\ell} (s_i^{n_\ell},\sigma_{-i}^{n_\ell}(\theta_{-i})) +\epsilon 
\end{align*}
for all $\ell\in \mathbb{N}$. This implies that for all $\ell \in \mathbb{N}$,
\begin{align*}
\sum_{i=1}^N v_i(&g_\zeta^{n_\ell}(\sigma_i^{n_\ell}(\theta_i),\sigma_{-i}^{n_\ell}(\theta_{-i})));\theta_i) \\[-.5em]
&> \sum_{i=1}^N  v_i(g_\zeta^{n_\ell} (s_i^{n_\ell},\sigma_{-i}^{n_\ell}(\theta_{-i}));\theta_i)+\epsilon.
\end{align*}
Since $g_\zeta^{n_\ell} (s_i^{n_\ell},\sigma_{-i}^{n_\ell}(\theta_{-i}))\in \mathcal{Z}$ due to the single fault tolerance, this contradicts the asymptotic efficiency of $g_\zeta^n\circ \sigma^n$.
\end{proof}
{\upd
In practice, the result of Proposition~\ref{tho:general} is used as follows. First, the leader chooses $\epsilon>0$ to which he wishes to diminish the followers' incentives to misbehave. Second, the leader identifies $n_0\in\mathbb{N}$ satisfying the condition in Definition~\ref{defasymeff} by analyzing the asymptotically efficient sequence of decision rules to be implemented. Finally, the leader announces a mechanism $M^n$ as defined in Proposition~\ref{tho:general} with some $n\geq n_0$. As a result, no follower has more than $\epsilon$ incentive to deviate from the suggested algorithm.}

Proposition~\ref{tho:general} requires $M^n$ to be single fault tolerant, i.e., that the social decision made by $M^n$ be feasible even if one of the agents is misbehaving. 
In some applications (such as \emph{dual decomposition} algorithms to be considered in Section~\ref{secdd}), this requirement can be met by simply projecting the intermediate result onto the feasible set. 
However, this is not always possible (as we will see in Section \ref{secaverage} where the \emph{average consensus} algorithm over a cyber-physical system is considered).
To circumvent this difficulty, we need a sequence of tax rules $\{g_{\pi_i}^n\}$ that leads the algorithm to converge to the optimal solution even under the lack of single fault tolerance. Several ideas can be exploited.
The following result is useful when it is easy for the leader to observe $\mathrm{dist}(\cdot,\mathcal{Z})$ and the knowledge of the upper bound $\beta(n)$ on $\mathrm{dist}(g_\zeta^n(\sigma^n(\theta)),\mathcal{Z})$ is available \emph{a priori}.
Roughly speaking, it penalizes all followers if the expected convergence rate (to the feasible set) is not observed.

\begin{proposition} \label{Prop:Infeasible} Assume that $\mathcal{Z}$ is a closed set and $v_i(z,\theta_i)$, $\forall i$, is a continuous function of $z$ for all $\theta_i\in\Theta_i$. Let $\{M^n\}_{n\in\mathbb{N}}$, $M^n=(\sigma^n,\mathcal{S}^n,g^n)$, be a sequence of mechanisms such that (\textit{i})~$\{g_\zeta^n \circ \sigma^n\}_{n\in\mathbb{N}}$ is asymptotically efficient and (\textit{ii})~$\sup_{\theta\in\Theta} \mathrm{dist}(g_\zeta^n(\sigma^n(\theta)),\mathcal{Z})\leq \beta(n)$ for some sequence $\{\beta(n)\}_{n\in\mathbb{N}}$ with $\lim_{n\rightarrow \infty} \beta(n)=0$. If for any $s^n\in\mathcal{S}^n$, a payment rule is chosen as
\[
g_{\pi_i}^n(s^n)\!=\!\begin{cases}
\sum_{j\neq i}v_j(g_\zeta^n(s^n);\theta_j) & \text{ if } \mathrm{dist}(g_\zeta^n(s^n),\mathcal{Z})\leq \beta(n) \\
C_i & \text{ otherwise }
\end{cases}
\]
then, for sufficiently large $C_i$,  $\{M^n\}_{n\in\mathbb{N}}$ is asymptotically incentive compatible.
In particular, the following choice suffices
\begin{align*}
C_i=-\inf_{\theta_i\in\Theta_i}\inf_{z\in \mathcal{Z}'} v_i(z;\theta_i)+ \sup_{\theta\in\Theta}\sup_{z\in \mathcal{Z}'}\sum_{j=1}^N v_j(z;\theta_j),\;\forall i
\end{align*}
where $\mathcal{Z}'=\{g_\zeta^n(\sigma^n(\theta)):\forall \theta\in\Theta\}$. 
\end{proposition}
\ifdefined\SHORTVERSION
\begin{proof}
The tax rule belongs to the indirect Groves class (Definition \ref{defindirectgroves}). Proof can be found in \cite{Tanaka:arxiv}.
\end{proof}
\fi

\ifdefined\LONGVERSION
\begin{proof}
The tax rule belongs to the indirect Groves class (Definition \ref{defindirectgroves}). Proof can be found in Appendix \ref{apendprop2}.
\end{proof}
\fi


\begin{remark}
Practical usefulness of the notion of asymptotic incentive compatibility heavily depends on the computational complexity of the algorithm. For computationally hard problems, realistically there is no mechanism that reduces undesirable incentive to $\epsilon$ in polynomial time. In such cases, asymptotic incentive compatibility may not be a convincing reasoning to induce faithful behaviors of followers. However, the issue of computational complexity requires more problem-specific discussions, which is not our focus in this paper.
\end{remark}

\section{Faithful implementation of dual decomposition}
\label{secdd}
\subsection{Algorithm}
\newcommand{\x}{z}
\newcommand{\X}{\mathcal{Z}}

Recall the dual decomposition algorithm considered in Section \ref{secresourceallocation}. Based on the developments so far, we are now going to design a tax rule that incentivize followers to execute (\ref{ahu1}) faithfully. The idea is to design an asymptotically incentive compatible sequence of mechanisms $M^n=(\sigma^n,\mathcal{S}^n,g^n)$, parameterized by the number of iterations $n$.
Proposition \ref{tho:general} shows that tax rules attaining our goal are not unique, since the choice of $k_i$ is arbitrary. In this section, we employ a particular tax rule among them inspired by the VCG mechanism. {\upd As we will see in the sequel, this tax rule turns out to be a natural choice since it is intimately related to the notion of ``market-clearing prices."} 

The VCG mechanism is also called the \emph{pivot mechanism},
since the tax for the follower $i$ is calculated based on the degree to which his presence/absence changes the social cost \cite{RefWorks:83}.
For every $i\in\{1,\cdots, N\}$, consider the ``marginal optimization problem" $P_{-i}(x)$ defined by
\begin{align}
P_{-i}(x): \;\; & \min \; \sum_{j\neq i} v_j(z_j;\theta_j)  \label{marginalprob}\\
& \text{ s.t.} \; \sum_{j\neq i} R_jz_j=c-R_ix. \nonumber
\end{align}
Notice that $P_{-i}(x)$ is the optimization problem  (\ref{originalprob}) in which follower $i$'s allocation is fixed at $z_i=x$. If his absence means $z_i=0$, the optimal social cost in his absence is obtained by solving $P_{-i}(0)$.
Let $(z^*, \lambda^*)$ and $(z^{-i*}(x), \lambda^{-i*}(x))$ be the primal-dual optimal solution to (\ref{originalprob})  and (\ref{marginalprob}) respectively.
The VCG-like tax for the $i$-th follower is defined by
\begin{equation}
\label{pvcg}
p_i^{VCG}=\sum_{j\neq i} v_j(z_j^*;\theta_j)-\sum_{j\neq i} v_j(z_j^{-i*}(0);\theta_j).
\end{equation}
In Algorithm \ref{distalg:reducedtax}, we propose a VCG-like mechanism for the faithful implementation of the dual decomposition algorithm applied to the resource allocation problem (\ref{originalprob}). In the iteration, variables $(\hat{z}^k, \lambda^k)$ and $(\hat{z}^{-i,k}, \lambda^{-i,k})$ are intended to approximate  $(z^*, \lambda^*)$ and $(z^{-i*}{\upd (0)}, \lambda^{-i*}{\upd (0)})$ respectively. 

\begin{algorithm}[!t]
\small
\caption{Distributed VCG mechanism $M^n=(\sigma^n,\mathcal{S}^n,g^n)$ for dual decomposition algorithms.}
\begin{algorithmic}[1]
\label{distalg:reducedtax} 
\small
\OUTPUT Allocation decision $\zeta^n(\theta)$ and tax assignments $\pi^n(\theta)$
\STATE (L) The leader announces the following algorithm;
\STATE // Solve optimization problem (\ref{originalprob}) by dual decomposition;
\STATE (L) Initialize and broadcast $\lambda^0$;
\FOR{$k=1,\dots,n$}
\STATE (F) Find $\hat{\x}_i^k=\argmin_{\x_i}(v_i(\x_i;\theta_i)+(\lambda^{k-1})^{\top}R_i\x_i)$ and report the result to the leader;
\STATE (L) Update and broadcast $\lambda^{k}=\lambda^{k-1}+\gamma (R\hat{\x}^{k}-c);$
\ENDFOR
\STATE // Solve marginal problems $P_{-j}(0)$ for every $j$;
\FOR{$j=1,\dots,N$}
\STATE (L) Initialize and broadcast $\lambda^{-j,0}$;
\FOR{$k=1,\dots,n$}
\STATE (F) Every follower $i(\neq j)$ finds and reports \\ $\hspace{.3in}\hat{\x}_i^{-j,k}\hspace{-.03in}=\hspace{-.03in}\argmin_{\x_i}(v_i(\x_i;\theta_i)+(\lambda^{-j,k-1})^{\top} \hspace{-.04in} R_i\x_i);$
\STATE (L) Update and broadcast \\ $\hspace{.3in}\lambda^{-j,k}=\lambda^{-j,k-1}+\gamma^{-j} (\sum_{i\neq j}R_i\hat{\x}^{-j,k}-c)$;
\ENDFOR
\ENDFOR
\STATE // Compute social outcomes;
\STATE (L) Compute and broadcast $z=\text{Proj}(\hat{z}^n,\mathcal{Z})$ and  $z^{-j}=\text{Proj}(\hat{z}^{-j,n},\mathcal{Z}^{-j}) $ for every  $j=1,\dots, N$;
\STATE (F) Every follower $i$ computes $\hat{v}_i=v_i(z_i;\theta_i)$ and $\hat{v}_i^{-j}=v_i(z_i^{-j};\theta_i)$ for every $j\neq i$ and report them to the leader;
\STATE (L) Determine taxes $\pi_i^n(\theta)=\sum_{j\neq i} \hat{v}_j-\sum_{j\neq i} \hat{v}_j^{-i}$  and allocations $\zeta^n(\theta)=z$;
\end{algorithmic}
\end{algorithm} 

\begin{proposition} \label{prop:1}
Assume that $v_i(\cdot;\theta_i)$, $i=1,2,\dots,N$, are strictly convex for every $\theta_i\in\Theta_i$. 
Then the sequence of mechanisms $\{M^n\}_{n\in\mathbb{N}}$ provided in Algorithm~\ref{distalg:reducedtax} is asymptotically incentive compatible.
\end{proposition}

\begin{proof}
\ifdefined\SHORTVERSION
See~\cite{Tanaka:arxiv}.
\fi
\ifdefined\LONGVERSION
See Appendix~\ref{proof:prop:1}.
\fi
\end{proof}

Notice that lines~8--15 in Algorithm~\ref{distalg:reducedtax} are devoted to calculating $k_i(s^n_{-i})$ in Proposition \ref{tho:general}. This is just one example of such function (which is motivated by taxes in VCG mechanisms).  Removing the above mentioned lines from Algorithm~\ref{distalg:reducedtax} and setting some other values to $k_i(s^n_{-i})$ (e.g.,  $k_i(s^n_{-i})=0$) still attains asymptotic incentive compatibility.  To relax strict convexity assumption, one may alternatively consider the ADMM algorithm \cite{boyd2011} (distributed implementation of the augmented Lagrangian algorithm \cite{bertsekas1995}) in Algorithm~\ref{distalg:reducedtax}.
 
\begin{example}
\label{solconcept}
As mentioned earlier, indirect Groves mechanisms in general implement efficient decision rules in \emph{ex-post} Nash equilibria but not in dominant strategies.
As an example, consider $\min_{z_1,z_2} \sum_{i=1,2} v_i(z_i)$ s.t. $z_1=z_2$ where $v_1(z_1)=(z_1-1)^2$ and $v_2(z_2)=(z_2-2)^2$, and the dual decomposition algorithm (Algorithm \ref{distalg:reducedtax}) is used. Suppose that the second player chooses to act as $\hat{z}_2^k=\argmin_{z_2}(\hat{v}_2(z_2)+\lambda^{k-1}z_2)$ where $\hat{v}_2(z_2)=(z_2-3)^2$ instead of line 5 in the algorithm. 
If the first player executes the algorithm faithfully, $\hat{z}_1^k$ and $\hat{z}_2^k$ converge to $2$, and the first player's net cost converges to $v_1(2)+p_1=v_1(2)+v_2(2)=1$. However, if the first player deviates from the suggested algorithm and executes $\hat{z}_1^k=\argmin_{z_1}(\hat{v}_1(z_1)+\lambda^{k-1}z_1)$ where $\hat{v}_1(z_1)=z_1^2$ instead of line 5 in the algorithm, $\hat{z}_1^k$ and $\hat{z}_2^k$ converge to $1.5$, and the first player's net cost converges to $v_1(1.5)+p_1=v_1(1.5)+v_2(1.5)=0.5$. Hence, the first player is better off by not following the intended algorithm. This means that faithful execution is not a dominant strategy.
\end{example}

\subsection{Connection between VCG and clearing prices}
\label{secconnection}

\begin{figure}[t]
    \centering
    \includegraphics[width=\linewidth, bb=30 40 710 450]{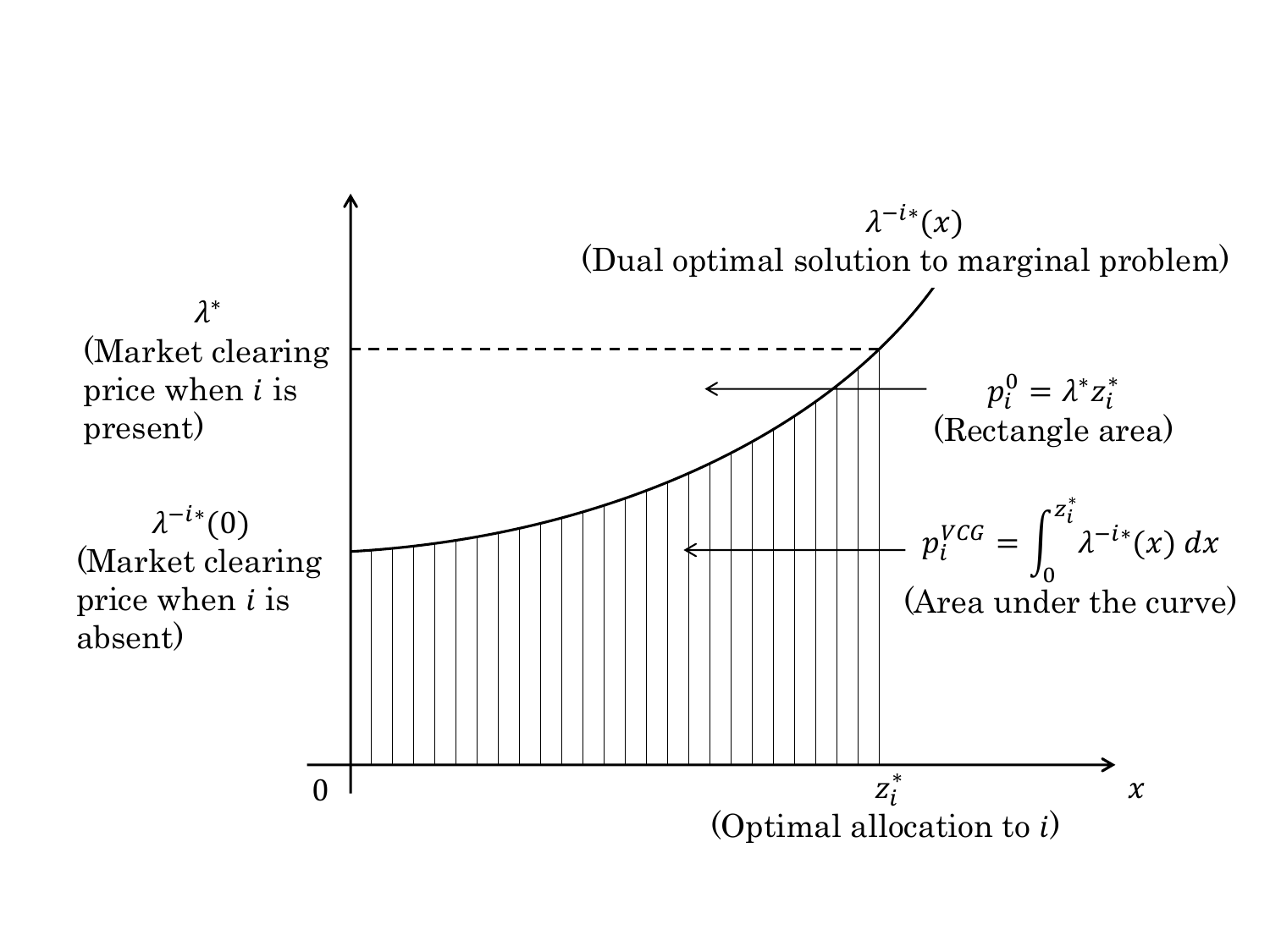}
    \caption{Graphical interpretation of the VCG taxes.}
    \label{fig:vcg}
\end{figure}

{\upd 
Invoking Theorem \ref{theoconverse}, it is now clear why the ``clearing price" mechanism considered in Section \ref{secresourceallocation} fails to be incentive compatible.
This is simply because the ``clearing price" mechanism is not in the class of indirect Groves mechanisms.
However, it can be shown that the ``clearing price" mechanism is \emph{approximately} incentive compatible under the pure competition \cite{RefWorks:238} (i.e., when individual followers have negligible market power to control market-clearing prices). We show this fact by pointing out an intimate connection between the tax rule $p_i^0=\lambda^*z_i^*$ considered in Section \ref{secresourceallocation} and the VCG-like tax in (\ref{pvcg}). For simplicity, we assume that cost functions $v_i$ are strongly convex and continuously differentiable.

To see a connection, for each $i$, consider a smooth path $x_i: [0,1]\rightarrow \mathbb{R}^{n_{z_i}}$ defined by $x_i(t)=tz_i^*$.}
Intuitively, the path $x_i$ continuously connects the follower $i$'s allocations in two distinct situations: $x_i(0)=0$ corresponds to the case where follower $i$ is absent (allocation is zero), and $x_i(1)=z_i^*$ is the optimal allocation in (\ref{originalprob}).
For every point on the path $x_i(t)$, it is possible to consider the primal-dual optimal solution $(z^{-i*}(x_i(t)), \lambda^{-i*}(x_i(t)))$ to the marginal optimization problem  (\ref{marginalprob}). In particular, $\lambda^{-i*}(x_i(t))$ shows how the market-clearing price determined by the rest of society (excluding $i$) changes if follower $i$'s allocation is fixed at different values $x_i(t)$ between $0$ and $z_i^*$. 

\begin{proposition}
\label{propmarginalvcg}
The VCG payment for the $i$-follower is obtained by integrating the market-clearing price $\lambda^{-i*}$ along the path $x_i$, i.e., $p_i^{VCG}=\int_0^{z_i^*}{\lambda^{-i*}}^{\top}(x_i)R_idx_i$.
\end{proposition}
\begin{proof}
For notational ease, dependency of $v_i$ on $\theta_i$ is suppressed.
Notice that
\begin{align}
\int_0^{z_i^*}&{\lambda^{-i*}}^{\top}(x_i)R_idx_i =\int_0^1 {\lambda^{-i*}}^{\top}(x_i(t))R_i\frac{dx_i(t)}{dt}dt \nonumber \\
&=-\sum_{j\neq i}\int_0^1  {\lambda^{-i*}}^{\top}(x_i(t))R_j\frac{dz_j^{-i*}(x_i(t))}{dt}dt \label{eq1} \\
&=\sum_{j\neq i} \int_0^1 \nabla_{z_j} v_j(z_j^{-i*}(x_i(t))) \frac{dz_j^{-i*}(x_i(t))}{dt}dt \label{eq2} \\
&=\sum_{j\neq i} \left[v_j(z_j^{-i*}(z_i^*))-v_j(z_j^{-i*}(0)) \right] \label{eq4}\\
&=\sum_{j\neq i} v_j(z_j^*)- \sum_{j\neq i}v_j(z_j^{-i*})=p_i^{VCG} \nonumber
\end{align}
The identity $R_i \frac{dx_i(t)}{dt}+\sum_{j\neq i}^N R_j\frac{dz_j^{-i*}(x_i(t))}{dt}=0$ is used in (\ref{eq1}) 
and (\ref{eq2}) is from the optimality condition
$
\nabla_{z_j} v(z_j^{-i*}(x))+{\lambda^{-i*}}^{\top}(x)R_j=0,\forall j\neq i.
$
The fundamental theorem of calculus is used in (\ref{eq4}). 
\end{proof}
A pictorial interpretation of $p_i^{VCG}$ is shown in Fig.  \ref{fig:vcg} (for simplicity, assume $R_i=1$ for every $i$ in the  resource allocation problem (\ref{originalprob}) in Section \ref{secresourceallocation}). 
{\upd
This figure shows that $p_i^{VCG} \approx p_i^0$ when the function $\lambda^{-i*}(x_i(\cdot))$ is nearly constant. 
In other words, when the individual followers have negligible market power, the ``clearing price" mechanism can be identified with the VCG mechanism.
}

\section{Faithful implementation of average consensus}
\label{secaverage}

{\upd Consider a multi-robot rendezvous problem in which $N$ robots want to meet in a single position (i.e., achieve a consensus in space). Considering that each robot utilizes fuel/battery to reach the rendezvous point, the social planner may want them to end up at a point that minimizes the sum of their distances from starting points, i.e., the average of the initial positions denoted by $\theta_i$. Here, we design a mechanism that can be used by the social planner (leader) to coordinate robots (followers) to faithfully implement a distributed algorithm, which lead them to the rendezvous point. 
In particular, we consider an iterative \emph{average consensus} algorithm~\cite{ren2005survey} in which robots are required to communicate with neighboring robots in each iteration.
Notice that an incentive design is needed in such situations, since otherwise a particular robot may choose to stand stationary, hoping that all other robots will move towards it thereby, not using any fuel. Due to the nature of the average consensus algorithms, the robots do not achieve exact consensus in their position in finite time. 
Moreover, if we consider the position $\zeta^n(\theta)$ of robots after $n$ iterations of communication as  the social decision, the decision rule $\zeta^n(\theta)$ with finite $n$ may not yield a feasible solution.
Hence the mechanism cannot be made single fault tolerant, and Proposition \ref{tho:general} is not applicable.
Instead, we use the result of Proposition~\ref{Prop:Infeasible}.
Note that the existence of a leader does not mean that communications are required between the leader and the followers at every iteration, and hence does not ruin the advantage of distributed consensus algorithms.
Indeed, in our design, a consensus is formed solely by local communications among robots, and the leader plays its role only at the beginning (to announce the tax rule and the algorithm to be implemented) and at the end (to compute taxes based on the information about the final positions of robots). 

In what follows, we formally introduce a mechanism that asymptotically implements the average consensus algorithm.}
Let an undirected graph $\mathcal{G}=(\{1,\dots,N\},\mathcal{E})$, with vertex set $\{1,\dots,N\}$ and edge set $\mathcal{E}$, be given to illustrate the communication links between the agents. 
Following~\cite{rabbat2005generalized}, we can achieve average consensus by solving 
\begin{subequations} \label{optimization:1}
\begin{align}
\min_{\x\in\mathbb{R}^N} & \sum_{i=1}^N (\x_i-\theta_i)^2,\\[-.3em] 
\mathrm{s.t.} \hspace{.05in} & \x_i=\x_j, \forall (i,j)\in\mathcal{E}, \label{optimization:1_2}
\end{align}
\end{subequations}
where $\x_i\in\mathbb{R}$ is the decision variable of follower~$i$ and $\theta_i\in\Theta_i\subseteq\mathbb{R}$ is its type. Let us define the incidence matrix $B\in\{-1,0,+1\}^{N\times |\mathcal{E}|}$ so that $b_{ij}=1$ if $e_j\in\mathcal{E}$ leaves vertex $i$, $b_{ij}=-1$ if $e_j\in\mathcal{E}$ enters vertex $i$, and $b_{ij}=0$ otherwise (assignments of directions to edges are arbitrary).

\begin{proposition}
\label{prop:consensus} Let $\mathcal{G}$ be a tree. The sequence of mechanisms $\{{M}^n\}_{n\in\mathbb{N}}$ provided in Algorithm~\ref{alg:3} is asymptotically incentive compatible.
\end{proposition}
\ifdefined\SHORTVERSION
\begin{proof}
The proof follows from applying Proposition~\ref{Prop:Infeasible}. For a detailed derivation of the rate $\{\beta(n)\}_{n\in\mathbb{N}}$ see~\cite{Tanaka:arxiv}.
\end{proof}
\fi
\ifdefined\LONGVERSION
\begin{proof}
Proof can be found in Appendix \ref{apendconsensus}.
\end{proof}
\fi

\begin{algorithm}[t]
\small
\caption{\label{alg:3} Distributed mechanism ${M}^n=(\sigma^n,\mathcal{S}^n,g^n)$ for asymptotically implementing the average consensus.}
\begin{algorithmic}[1]
\OUTPUT Consensus decision $\zeta^n(\theta)$ and tax assignment $\pi^n(\theta)$
\STATE (L) The leader announces the following algorithm;
\STATE //  Solve optimization problem~\eqref{optimization:1};
\STATE (L) Set $\alpha\in(0,1/\lambda_{\max}(B^\top B))$  and broadcast it;
\STATE (L) Set $\rho=\lambda_{\min}(B^\top B)/ \lambda_{\max}(B^\top B)$;
\STATE (F) Initialize $z_i^0=\theta_i$ for each $1\leq i\leq N$;
\FOR{$k=1,\dots,n$}
\STATE (F) Calculate $z_i^k=z_i^{k-1}+\alpha \sum_{j\in\mathcal{N}_i} (z_j^{k-1} -z_i^{k-1})$, where $\mathcal{N}_i$ is the neighbors of agent~$i$ in~$\mathcal{G}$, and transmit to neighbors;
\ENDFOR
\STATE (F) Transmit $z_i^n$ and $\hat{v}_i=(z_i^n-\theta_i)^2$ to the leader;
\STATE // Compute social outcomes;
\STATE (L) Determine and broadcast $\zeta^n(\theta)=z^n$;
\IF{$\mathrm{dist}(z^n,\mathcal{Z})\leq(1-\rho)^{n} \|B(B^{\top}B)^{-1}\|_{2}\sup_{q\in\Theta}\|B^{\top}q\|_2$}
\STATE (L) Determine $\pi_i^n(\theta)=\sum_{j\neq i}\hat{v}_j$; 
\ELSE
\STATE (L) Determine $\pi_i^n(\theta)=\sup_{\theta\in\Theta}\sup_{z\in \mathcal{Z}'}\sum_{j} (z_j-\theta_j)^2$, where $\mathcal{Z}'=\mathcal{Y}^N$ with $\mathcal{Y}=\{\sum_{j}\alpha_j\theta_j: \sum_{j} \alpha_j=1,\alpha_j\geq 0,\theta_j\in\Theta_j,\forall j\}. $
\ENDIF
\end{algorithmic}
\end{algorithm}


\section{Discussion and Conclusions}
\label{secdiscussion}
We have discussed a general indirect mechanism design framework for faithful distributed algorithm implementations. As examples, we have considered dual decomposition and average consensus algorithms. 

{\upd
The framework of this paper is directly applicable to many \emph{distributed control problems}. Although the issue of incentive is usually neglected in the control theory literature, it is an important challenge that always arises when distributed agents are strategic. For example, a distributed control algorithm proposed in \cite{RefWorks:239} assumes that agents are faithful to the algorithm. However, this requirement can be removed by introducing the tax mechanism in Algorithm~\ref{distalg:reducedtax}.

In the future, we may also consider the faithful implementations of \emph{distributed model predictive control} (MPC) algorithms~\cite{camponogara2002distributed}. 
Distributed MPC is expected to be a powerful tool in large-scale social engineering problems (e.g., operations of power systems \cite{venkat2008} and transportation systems \cite{negenborn2008multi}).
Faithful implementations of distributed MPC requires online (real-time) mechanisms. Several appropriate modifications need to be made to the current framework (e.g., replacement of the solution concept from Nash equilibrium to Markov perfect equilibrium). Further study will be required in this research direction. 
We believe this is a great opportunity for economic theory (i.e., mechanism design) and engineering (i.e., control) to merge in order to tackle challenging problems in the society.}



\appendix

\subsection{Proof of Theorem \ref{theoconverse}}
\label{apendnec}
\underline{Case 1}:  We will first show that the tax rule must be in the form of (\ref{grovestaxequation}) when $s_i \in \sigma_i(\Theta_i)$. Proof is by contradiction.
Suppose  that 
$ g_{\pi_i}(s_i, \sigma_{-i}(\theta_{-i})) =\sum_{j \neq i} v_j (g_\zeta(s_i, \sigma_{-i}(\theta_{-i}));\theta_j) + k_i(s_i, \sigma_{-i}(\theta_{-i}))  $
where there exist $s_i, s_i' \in \sigma_i(\Theta_i)$ such that
\begin{equation}
\label{kineq}
k_i(s_i, \sigma_{-i}(\theta_{-i}))\neq k_i(s_i', \sigma_{-i}(\theta_{-i})).
\end{equation}

(Step 1): Suppose $g_\zeta(s_i, \sigma_{-i}(\theta_{-i}))=g_\zeta(s_i', \sigma_{-i}(\theta_{-i}))$. Let $\theta_i,\theta_i'\in \Theta_i$ satisfy  $s_i=\sigma_i(\theta_i)$ and $s_i'=\sigma_i(\theta_i')$. If player $i$'s true parameter is $\theta_i$, {\upd acting $s_i$ minimizes his net cost since $M$ is incentive compatible.} Thus,
\begin{align*}
& v_i(g_\zeta(s_i, \sigma_{-i}(\theta_{-i})); \theta_i)+g_{\pi_i}(s_i, \sigma_{-i}(\theta_{-i})) \\
& \leq v_i(g_\zeta(s_i', \sigma_{-i}(\theta_{-i})); \theta_i)+g_{\pi_i}(s_i', \sigma_{-i}(\theta_{-i})) 
\end{align*}
or
$g_{\pi_i}(s_i, \sigma_{-i}(\theta_{-i})) \leq g_{\pi_i}(s_i', \sigma_{-i}(\theta_{-i}))$.
On the other hand, if player $i$'s true parameter is $\theta_i'$, it must be that 
$
g_{\pi_i}(s_i, \sigma_{-i}(\theta_{-i})) \geq g_{\pi_i}(s_i', \sigma_{-i}(\theta_{-i})).
$
Since  both cases could occur, the only possibility is 
\begin{equation}
\label{eqgpis}
g_{\pi_i}(s_i, \sigma_{-i}(\theta_{-i})) = g_{\pi_i}(s_i', \sigma_{-i}(\theta_{-i})).
\end{equation}
The left hand side of equation (\ref{eqgpis}) can be written as $\sum_{j\neq i} v_j(g_\zeta(s_i, \sigma_{-i}(\theta_{-i}));\theta_j) + k_i(s_i, \sigma_{-i}(\theta_{-i}))$, while the right hand side is $\sum_{j\neq i} v_j(g_\zeta(s_i', \sigma_{-i}(\theta_{-i}));\theta_j) + k_i(s_i', \sigma_{-i}(\theta_{-i}))$.
Since we are assuming $g_\zeta(s_i, \sigma_{-i}(\theta_{-i}))=g_\zeta(s_i', \sigma_{-i}(\theta_{-i}))$, equation (\ref{eqgpis}) {\upd implies}
\[ k_i(s_i, \sigma_{-i}(\theta_{-i})) = k_i(s_i', \sigma_{-i}(\theta_{-i})).\]
This is a contradiction to (\ref{kineq}).

(Step 2): Now we can assume $g_\zeta(s_i, \sigma_{-i}(\theta_{-i})) \neq g_\zeta(s_i', \sigma_{-i}(\theta_{-i}))$. Without loss of generality, we can also assume $ k_i(s_i, \sigma_{-i}(\theta_{-i})) < k_i(s_i', \sigma_{-i}(\theta_{-i})) $. Then,
\begin{equation}
\label{ineqeps}
k_i(s_i, \sigma_{-i}(\theta_{-i})) < k_i(s_i', \sigma_{-i}(\theta_{-i}))-\epsilon
\end{equation}
for some $\epsilon>0$.
For every $j\neq i,  1 \leq j \leq  N$, assume $v_j(\cdot; \theta_j)$ are {\upd  quadratic functions}.
Since $\Theta_i$ exhausts the space of all quadratic functions,  there exists $\theta_i''\in \Theta_i$ such that
\begin{align}
&v_i(z;\theta_i'')=-\sum_{j\neq i} v_j(z;\theta_j) \nonumber \\
&+\epsilon \left( \frac{\|z-g_\zeta(s_i', \sigma_{-i}(\theta_{-i}))\|^2}{\|g_\zeta(s_i, \sigma_{-i}(\theta_{-i}))-g_\zeta(s_i', \sigma_{-i}(\theta_{-i})) \|^2}-1\right).\label{viconst}
\end{align}
By incentive compatibility, player $i$ whose true type is $\theta_i''$ {\upd attains smaller net cost} by acting $\sigma_i(\theta_i'')$ than acting $s_i$. (Here, $\sigma_i(\theta_i'')$ and $s_i$ may or may not be equal.)
\begin{align}
&v_i(g_\zeta(\sigma_i(\theta_i''),\sigma_{-i}(\theta_{-i}));\theta_i'')+g_{\pi_i}(\sigma_i(\theta_i''),\sigma_{-i}(\theta_{-i}))  \nonumber \\
& \leq  v_i(g_\zeta(s_i,\sigma_{-i}(\theta_{-i}));\theta_i'')+g_{\pi_i}(s_i,\sigma_{-i}(\theta_{-i})). \label{eqstar2}
\end{align}
By efficiency, {\upd the outcome $z$ of the mechanism minimizes}
\begin{align*}
&v_i(z;\theta_i'')+\sum_{j\neq i} v_j(z;\theta_j) \\
&=\epsilon \left( \frac{\|z-g_\zeta(s_i', \sigma_{-i}(\theta_{-i}))\|^2}{\|g_\zeta(s_i, \sigma_{-i}(\theta_{-i}))-g_\zeta(s_i', \sigma_{-i}(\theta_{-i})) \|^2}-1\right).
\end{align*}
Thus, {\upd the decision by the mechanism is}
\begin{equation}
\label{effchoice}
g_\zeta(\sigma_i(\theta_i''),\sigma_{-i}(\theta_{-i}))=g_\zeta(s_i',\sigma_{-i}(\theta_{-i})).
\end{equation}
Substituting (\ref{effchoice}) into (\ref{eqstar2}) gives
\begin{align*}
&v_i(g_\zeta(s_i', \sigma_{-i}(\theta_{-i}));\theta_i'')+g_{\pi_i}(\sigma_i(\theta_i''),\sigma_{-i}(\theta_{-i})) \\
 &\leq 
v_i(g_\zeta(s_i, \sigma_{-i}(\theta_{-i}));\theta_i'')+g_{\pi_i}(s_i,\sigma_{-i}(\theta_{-i})).
\end{align*}
Rewriting the above relation using (\ref{viconst}),
\begin{align*}
&-\sum_{j\neq i} v_j(g_\zeta(s_i',\sigma_{-i}(\theta_{-i}));\theta_j)\!-\!\epsilon 
\! \\
&+ \sum_{j\neq i} v_j(g_\zeta(\sigma_i(\theta_i''),\sigma_{-i}(\theta_{-i}));\theta_j)\!+\!k_i(\sigma_i(\theta_i''),\sigma_{-i}(\theta_{-i})) \\
&\leq 
-\sum_{j\neq i} v_j(g_\zeta(s_i,\sigma_{-i}(\theta_{-i}));\theta_j) 
\\
&+ \sum_{j\neq i} v_j(g_\zeta(s_i,\sigma_{-i}(\theta_{-i}));\theta_j)+k_i(s_i,\sigma_{-i}(\theta_{-i})).
\end{align*}
Using (\ref{effchoice}) again, this can be simplified to 
\begin{equation}
\label{eqprA}
k_i(\sigma_i(\theta_i''),\sigma_{-i}(\theta_{-i}))-\epsilon \leq k_i(s_i,\sigma_{-i}(\theta_{-i})).
\end{equation}
Since (\ref{effchoice}), by applying the argument in Case 1, it must follow that
\begin{equation}
\label{eqprB}
k_i(\sigma_i(\theta_i''),\sigma_{-i}(\theta_{-i})) = k_i(s_i',\sigma_{-i}(\theta_{-i})).
\end{equation}
Substituting (\ref{eqprB}) into (\ref{eqprA}), we have $k_i(s_i',\sigma_{-i}(\theta_{-i})) -\epsilon \leq k_i(s_i,\sigma_{-i}(\theta_{-i}))$. This is a contradiction to (\ref{ineqeps}). Hence, we have shown that the tax rule must be in the form of (\ref{grovestaxequation}) when $s_i \in \sigma_i(\Theta_i)$.

\underline{Case 2}: Now we need to show that the inequality (\ref{grovestaxineq}) must be satisfied for every $s_i\in \mathcal{S}_i \setminus \sigma(\Theta_i)$, where $k_i$ is the same function as in Case 1. Suppose, on the contrary, that there exist $s_i\in \mathcal{S}_i \setminus \sigma(\Theta_i)$ and $\theta_{-i}\in \Theta_{-i}$ such that
\begin{align}
&g_{\pi_i}(s_i,\sigma_{-i}(\theta_{-i})) + \epsilon \nonumber  \\
& <  \sum_{j\neq i}  v_j(g_\zeta(s_i,\sigma_{-i}(\theta_{-i}));\theta_j)+k_i(\sigma_{-i}(\theta_{-i})) \label{case2eps}
\end{align}
with $\epsilon>0$. By incentive compatibility, the $i$-th follower
{\upd can minimize the net cost by acting $\sigma_i(\theta_i)$}:
\begin{align*}
v_i(g_\zeta(\sigma_i(\theta_i),\sigma_{-i}(\theta_{-i}));\theta_i)+g_{\pi_i}(\sigma_i(\theta_i),\sigma_{-i}(\theta_{-i}))  \\
\leq v_i(g_\zeta(s_i,\sigma_{-i}(\theta_{-i}));\theta_i)+g_{\pi_i}(s_i,\sigma_{-i}(\theta_{-i})). 
\end{align*}
By the discussion in Case 1, the equality $g_{\pi_i}(\sigma_i(\theta_i),\sigma_{-i}(\theta_{-i}))=  \sum_{j\neq i}  v_j(g_\zeta(\sigma_i(\theta_i),\sigma_{-i}(\theta_{-i}));\theta_j)+k_i(\sigma_{-i}(\theta_{-i})) $ is applicable on the left hand side. Also, by substituting (\ref{case2eps}) into the right hand side, we obtain
\[\sum_{i=1}^N\! v_i(g_\zeta (\sigma_i(\theta_i),\sigma_{-\!i}(\theta_{-\!i}));\!\theta_i) \!<\!\! \sum_{i=1}^N\! v_i(g_\zeta (s_i,\sigma_{-\!i}(\theta_{-\!i}));\!\theta_i)-\epsilon.\]
Now, consider an extreme situation in which all cost functions are constant, i.e., $v_i(\cdot;\theta_i)=c$ for every $i\in\{1,\cdots, N\}$. Then, the last inequality leads to $\epsilon < 0$, a contradiction.

\ifdefined\LONGVERSION
\subsection{Proof of Proposition \ref{Prop:Infeasible}}
\label{apendprop2}
To prove this proposition, assume that, $\forall n\in\mathbb{N}$, agent $i$ follows $s_i^n\in\mathcal{S}_i^n$ and the rest of the agents follow $\sigma_{-i}^n(\theta_{-i})$.
Let us define sets 
\begin{align*}
\mathcal{M}&=\{m\in\mathbb{N}\,|\,\mathrm{dist}(g_\zeta^{n_m}(s_{i}^{n_m},\sigma_{-i}^{n_m}(\theta_{-i})),\mathcal{Z})> \beta(n_m)\},\\
\mathcal{L}&=\{\ell\in\mathbb{N}\,|\,\mathrm{dist}(g_\zeta^{n_\ell}(s_{i}^{n_\ell},\sigma_{-i}^{n_\ell}(\theta_{-i})),\mathcal{Z})\leq \beta(n_\ell)\}.
\end{align*}
Now, we prove the following two claims.

\underline{Claim 1}: If $\mathcal{M}\neq \emptyset$, $\forall\epsilon>0$, $\exists\bar{m}\in\mathcal{M}$ such that 
\begin{align*} 
&v_i(g_\zeta^{n_m}(\sigma_{i}^{n_m}(\theta_i),\sigma_{-i}^{n_m}(\theta_{-i}));\theta_i)+ g_{\pi_i}^{n_m}(\sigma_i^{n_m}(\theta_i),\sigma_{-i}^{n_m} (\theta_{-i}))\nonumber\\ &\leq v_i(g_\zeta^{n_m}(s_{i}^{n_m},\sigma_{-i}^{n_m}(\theta_{-i}));\theta_i) + g_{\pi_i}^{n_m}(s_i^{n_m},\sigma_{-i}^{n_m} (\theta_{-i}))+\epsilon,
\end{align*}
for all $m\in\mathcal{M}$ such that $m\geq \bar{m}$.

To prove this claim, note that, $\forall m\in\mathcal{M}$, we have
\begin{align}
v_i(&g_\zeta^{n_m}(s_{i}^{n_m},\sigma_{-i}^{n_m}(\theta_{-i}));\theta_i)+g_{\pi_i}^{n_{m}}(s_i^{n_{m}},\sigma_{-i}^{n_{m}} (\theta_{-i})) \nonumber\\
=&v_i(g_\zeta^{n_m}(s_{i}^{n_m},\sigma_{-i}^{n_m}(\theta_{-i}));\theta_i)\nonumber\\
&-\inf_{\theta_i\in\Theta_i}\inf_{z\in \mathcal{Z}'} v_i(z;\theta_i) + \sup_{\theta\in\Theta}\sup_{z\in \mathcal{Z}'}\sum_{j=1}^N v_j(z;\theta_j)\nonumber\\
\geq& \sup_{\theta\in\Theta}\sup_{z\in \mathcal{Z}'}\sum_{j=1}^N v_j(z;\theta_j).
\label{inequality:aux:1}
\end{align}
Now, because $g_\zeta^{n_m}(\sigma_{i}^{n_m}(\theta_i),\sigma_{-i}^{n_m}(\theta_{-i}))\in\mathcal{Z}'$, we get
\begin{align*}
\sup_{\theta\in\Theta}\sup_{z\in \mathcal{Z}'}\sum_{j=1}^N v_j(z;\theta_j)\geq \sum_{j=1}^N v_j(g_\zeta^{n_m}(\sigma_i^{n_m}(\theta_i),\sigma_{-i}^{n_m}(\theta_{-i}));\theta_j),
\end{align*}
which, in combination with~\eqref{inequality:aux:1}, gives
\begin{align*}
&v_i(g_\zeta^ {n_m}(\sigma_{i}^{n_m}(\theta_{i}),\sigma_{-i}^{n_m}(\theta_{-i}));\theta_i)+ g_{\pi_i}^{n_{m}}(\sigma_i^{n_{m}}(\theta_i),\sigma_{-i}^{n_{m}} (\theta_{-i})) \\
&\leq v_i(g_\zeta^{n_m}(s_{i}^{n_m},\sigma_{-i}^{n_m}(\theta_{-i}));\theta_i)+g_{\pi_i}^{n_{m}}(s_i^{n_{m}},\sigma_{-i}^{n_{m}} (\theta_{-i})).
\end{align*}
This proves Claim~1 by setting $\bar{m}=\min_{m\in\mathcal{M}}m$ (which is well-defined as $\mathcal{M}\neq \emptyset$).

\underline{Claim~2}: If $\mathcal{L}\neq \emptyset$ and $|\mathcal{L}|=\infty$, $\forall\epsilon>0$, $\exists\bar{\ell}\in\mathcal{L}$ such that 
\begin{align*} 
&v_i(g_\zeta^{n_\ell}(\sigma_{i}^{n_\ell}(\theta_i),\sigma_{-i}^{n_\ell}(\theta_{-i}));\theta_i)+ g_{\pi_i}^{n_\ell}(\sigma_i^{n_\ell}(\theta_i),\sigma_{-i}^{n_\ell} (\theta_{-i}))\nonumber\\ &\leq v_i(g_\zeta^{n_\ell}(s_{i}^{n_\ell},\sigma_{-i}^{n_\ell}(\theta_{-i}));\theta_i) + g_{\pi_i}^{n_\ell}(s_i^{n_\ell},\sigma_{-i}^{n_\ell} (\theta_{-i}))+\epsilon,
\end{align*}
for all $\ell\in\mathcal{L}$ such that $\ell\geq \bar{\ell}$.

Let us define 
$$
\hat{z}^{n_\ell}\in\argmin_{z\in\mathcal{Z}} \; \|g_\zeta^{n_\ell}  (s_{i}^{n_\ell},\sigma_{-i}^{n_\ell}(\theta_{-i}))-z\|_2,\;\; \forall \ell\in\mathcal{L}.
$$
Let us show $\argmin_{z\in\mathcal{Z}} \|g_\zeta^{n_\ell}  (s_{i}^{n_\ell},\sigma_{-i}^{n_\ell}(\theta_{-i}))-z\|_2\neq \emptyset$. Fix an arbitrary $z_0\in\mathcal{Z}$ and define $\mathcal{B}=\{z\in\mathbb{R}^{n_z}|\|z-z_0\|_2\leq 2 \|g_\zeta^{n_\ell} (s_{i}^{n_\ell},\sigma_{-i}^{n_\ell}(\theta_{-i}))-z_0\|_2\}$. For all $z\in\mathcal{Z}\setminus \mathcal{B}$, we have
\begin{align*}
\|g_\zeta^{n_\ell} (s_{i}^{n_\ell},\sigma_{-i}^{n_\ell}(\theta_{-i}))-z\|_2
\geq &\|z-z_0\|_2\nonumber\\
&-\|g_\zeta^{n_\ell} (s_{i}^{n_\ell},\sigma_{-i}^{n_\ell}(\theta_{-i}))-z_0\|_2\nonumber\\
>&\|g_\zeta^{n_\ell} (s_{i}^{n_\ell},\sigma_{-i}^{n_\ell}(\theta_{-i}))-z_0\|_2,
\end{align*}
where the first and the second inequalities, respectively, follow from the triangular inequality and that $z\in\mathcal{Z}\setminus \mathcal{B}$. Hence,
\begin{align}
\inf_{z\in\mathcal{Z}} \|g_\zeta^{n_\ell} & (s_{i}^{n_\ell},\sigma_{-i}^{n_\ell}(\theta_{-i}))-z\|_2
\nonumber\\&=\inf_{z\in\mathcal{Z}\cap\mathcal{B}} \|g_\zeta^{n_\ell}  (s_{i}^{n_\ell},\sigma_{-i}^{n_\ell}(\theta_{-i}))-z\|_2, \label{eqnn:infimum}
\end{align}
because, as we showed above, any $z\notin\mathcal{Z}\cap\mathcal{B}$ results in a strictly larger distance than $z_0\in\mathcal{Z}\cap\mathcal{B}$. Note that $\mathcal{Z}\cap \mathcal{B}$ is compact because it is bounded (subset of bounded set $\mathcal{B}$) and closed (intersection of two closed sets). Theorem~4.16 in~\cite[p.\,89]{rudin1976principles} shows $\exists z'\in \mathcal{Z}\cap\mathcal{B}$ that achieves the infimum on the right hand side of~\eqref{eqnn:infimum} and, thus, the infimum on the left hand side of~\eqref{eqnn:infimum}. Therefore, $z'\in\argmin_{z\in\mathcal{Z}} \|g_\zeta^{n_\ell} (s_{i}^{n_\ell},\sigma_{-i}^{n_\ell}(\theta_{-i}))-z\|_2$ because $z'\in\mathcal{Z}$. Therefore, $\hat{z}^{n_\ell}$ is well-defined. For any $\epsilon>0$, there exists $\ell_1\in\mathcal{L}$ such that for all $\ell\in\mathcal{L}$ that $\ell\geq \ell_1$, we get
\begin{equation*}
\begin{split}
\bigg|\sum_{j=1}^N v_j(g_\zeta^{n_\ell}(s_{i}^{n_\ell},\sigma_{-i}^{n_\ell}(\theta_{-i}));\theta_j)
-\sum_{j=1}^N v_j(\hat{z}^{n_\ell};\theta_j)\bigg|\leq \epsilon/2,
\end{split}
\end{equation*}
because of the continuity of $v_i(\cdot;\theta_i)$, $\forall i$, and the fact that $\lim_{\ell\in\mathcal{L},\ell\rightarrow \infty} \mathrm{dist}(g_\zeta^{n_\ell}(s_{i}^{n_\ell},\sigma_{-i}^{n_\ell}(\theta_{-i})),\mathcal{Z})=0$ since $\{\beta(n_\ell)\}_{\ell\in\mathcal{L}}$ is a vanishing sequence (as $|\mathcal{L}|=\infty$). Hence, 
\begin{align} 
\sum_{j=1}^N v_j(\hat{z}^{n_\ell};\theta_j)\leq & \sum_{j=1}^N v_j(g_\zeta^{n_\ell}(s_{i}^{n_\ell},\sigma_{-i}^{n_\ell}(\theta_{-i}));\theta_j)+\epsilon/2 \nonumber\\ = & v_i(g_\zeta^{n_\ell}(s_{i}^{n_\ell},\sigma_{-i}^{n_\ell}(\theta_{-i}));\theta_i) \nonumber\\ & \hspace{.25in}+ g_{\pi_i}^{n_\ell}(s_i^{n_\ell},\sigma_{-i}^{n_\ell} (\theta_{-i}))+\epsilon/2.\label{eqn:proof:4}
\end{align}
Note that, by construction, $\hat{z}^{n_\ell}\in\mathcal{Z}$, $\forall\ell\in\mathcal{L}$. Hence, for any $\epsilon>0$, there exists $\ell_2\in\mathcal{L}$ such that for all $\ell\in\mathcal{L}$ that $\ell\geq \ell_2$, we get
\begin{align} 
v_i(g_\zeta^{n_\ell}&(\sigma_{i}^{n_\ell}(\theta_i),\sigma_{-i}^{n_\ell}(\theta_{-i}));\theta_i)+ g_{\pi_i}^{n_\ell}(\sigma_i^{n_\ell}(\theta_i),\sigma_{-i}^{n_\ell} (\theta_{-i})) \nonumber\\
&=\sum_{j=1}^N v_j(g_\zeta^{n_\ell}(\sigma_{i}^{n_\ell}(\theta_i),\sigma_{-i}^{n_\ell}(\theta_{-i}));\theta_j) \nonumber\\
&\leq \sum_{j=1}^N v_j(\hat{z}^{n_\ell};\theta_j)+\epsilon/2,
\label{eqn:proof:5}
\end{align}
where the last inequality is because $\{g_\zeta^{n_\ell} \circ \sigma^{n_\ell}\}_{\ell\in\mathcal{L}}$ is asymptotically efficient as $|\mathcal{L}|=\infty$ (note that, by definition, if a sequence of decision rules is asymptotically efficient, every infinite subsequence of it is also asymptotically efficient). Combining~\eqref{eqn:proof:4} and~\eqref{eqn:proof:5} while setting $\bar{\ell}=\max(\ell_1,\ell_2)$ proves Claim~2.

Now, we are ready to prove the statement of this proposition. If $|\mathcal{L}|<\infty$ (which implies that $|\mathcal{M}|=\infty$ as, by definition, $\mathcal{M}\cup\mathcal{L}=\mathbb{N}$ and $\mathcal{M}\cap\mathcal{N}=\emptyset$), the proposition follows from Claim~1 and by setting $n_0=\max\{n_{\max_{\ell\in\mathcal{L}}\ell},n_{\bar{m}}\}$ (in Definition~\ref{defaic}). Otherwise, if $|\mathcal{L}|=\infty$, the proof follows from Claims~1 and~2 and by setting $n_0=\max\{n_{\bar{\ell}},n_{\bar{m}}\}$.
\fi

\ifdefined\LONGVERSION
\subsection{Proof of Proposition~\ref{prop:1}} \label{proof:prop:1}
First, note that the payments introduced in Algorithm~\ref{distalg:reducedtax} is of the form introduced in Proposition~\ref{tho:general} when setting $k_i(s_{-i}^n)=-\sum_{j\neq i} \hat{v}_j^{-i,n}$ for all $1\leq i\leq N$. Note that this quantity does not depend on follower $i$'s actions. {\upd Second, for each $n\in\mathbb{N}$, the mechanism $M^n$ is single fault tolerant, since the outcome of the mechanism is guaranteed to be feasible (i.e., $\zeta^n(\theta)\in\mathcal{Z}$) due to the operation in line~17 of Algorithm~\ref{distalg:reducedtax}.}
Third, by the convergence property of the dual decomposition algorithm and the continuity of {\upd cost functions, we have
\begin{align*}
\lim_{n\rightarrow \infty} \sum_{i=1}^N v_j(\zeta_i^n(\theta);\theta_i)=\min_{Rz=c} \sum_{i=1}^N v_i(z_i;\theta_i).
\end{align*}
Hence, the sequence of decision rules $\{\zeta^n\}_{n\in\mathbb{N}}$ is asymptotically efficient.} Now, the rest follows from  Proposition~\ref{tho:general}.
\fi

\ifdefined\LONGVERSION
\subsection{Proof of Proposition \ref{prop:consensus}}
\label{apendconsensus}
Note that $\mathcal{Z}=\{z|B^\top z=0\}$. For any $\hat{z}\in\mathbb{R}^N$, $\mathrm{dist}(\hat{z},\mathcal{Z})=\|B(B^{\top}B)^{-1}B^\top\hat{z}\|_2$ as $(I-B(B^{\top}B)^{-1}B^{\top})\hat{z}$ is the projection of $\hat{z}$ into $\mathcal{Z}$. Following~\cite[p.\,74]{lutkepohl1996handbook}, we get
\begin{align*}
\lambda_{\max}(I-\alpha B^{\top}B)
&\leq \lambda_{\max}(I)+\lambda_{\max}(-\alpha B^{\top}B)\\
&=\lambda_{\max}(I)-\alpha\lambda_{\min}(B^{\top}B)\\
&\leq1-\rho,
\end{align*}
and
\begin{align*}
\lambda_{\min}(I-\alpha B^{\top}B)
&\geq \lambda_{\min}(I)+\lambda_{\min}(-\alpha B^{\top}B)\\
&=\lambda_{\min}(I)-\alpha\lambda_{\max}(B^{\top}B)\\
&\geq 0.
\end{align*}
Now, note that
\begin{equation*}
\begin{split}
B^{\top}z^{k+1}&=B^{\top}(I-\alpha BB^{\top})z^k=(I-\alpha B^{\top}B)B^{\top}z^k,
\end{split}
\end{equation*}
which results in
\begin{equation*}
\begin{split}
\|B^{\top}z^{k+1}\|_2^2&\leq\|I-\alpha B^{\top}B\|_2^2\|B^{\top}z^k\|_2^2\\
&\leq\lambda_{\max}(I-\alpha B^{\top}B)^2\|B^{\top}z^k\|_2^2\\
&\leq(1-\rho)^2\|Rz^k\|_2^2,
\end{split}
\end{equation*}
where the second inequality follows from~\cite[p.\,133]{lutkepohl1996handbook}. This shows that
\begin{equation*}
\begin{split}
\|B^{\top}z^k\|_2^2&\leq (1-\rho)^{2k}\|B^{\top}z^0\|_2^2\leq (1-\rho)^{2k}\sup_{q\in\Theta}\|B^{\top}q\|_2^2.
\end{split}
\end{equation*}
Define $\beta(n)=(1-\rho)^{n} \|B(B^{\top}B)^{-1}\|_2\sup_{q\in\Theta}\|B^{\top}q\|_2$. Evidently, $\lim_{n\rightarrow\infty} \beta(n)=0$ since $0<\rho\leq1$ (as $\lambda_{\min}(B^{\top}B)>0$ because $\mathcal{G}$ is a tree). Furthermore, since $\mathbf{1}^{\top}z^k=\mathbf{1}^{\top}z^0$ following the update dynamics in Algorithm~\ref{alg:3} (see~\cite{xiao2004fast}), we get $\lim_{k\rightarrow\infty} z^k=\left(\frac{1}{N}\sum_{i=1}^N \theta_i\right)\mathbf{1}.$ Therefore, $\{g_\zeta^n\circ \sigma^n\}_{n\in\mathbb{N}}$ is asymptotically efficient. Now, the rest of the proof follows from applying Proposition~\ref{Prop:Infeasible}.
\fi

\bibliographystyle{ieeetr}
\bibliography{Refs}

\end{document}